%% file: main.tex
\documentclass[10pt]{llncs}

\usepackage{algorithm}
\usepackage{algorithmic}
\usepackage{amssymb}
\usepackage{xspace}
\usepackage{url}

\newtheorem{observation}{Observation}

\usepackage{graphicx}
\usepackage{color}

\newcommand{\deps}{$\ \textsf{dep} \ $}

\begin{document}
\title{GentleRain+: Making GentleRain Robust on Clock Anomalies}

\author{Mohammad Roohitavaf \and Sandeep S. Kulkarni}

\institute{
Department of Computer Science and Engineering\\
Michigan State University\\
East Lansing, MI 48824, USA\\
Email: \url{\{roohitav,sandeep\}@cse.msu.edu}
}

\maketitle

\input{abstract} 
\input{introduction}

\input{background}

\input{motivation}
\input{hlc}
\input{protocol}
\input{results}

\input{impossibility}
\input{discussion}
\input{conclusion}
\bibliographystyle{unsrt}
\bibliography{bib}

\end{document}

%% file: abstract.tex
\begin{abstract}
 
Causal consistency is in an intermediate consistency model that can be achieved together with high availability and high performance requirements even in presence of network partitions. There are several proposals in the literature for causally consistent data stores. Thanks to the use of single scalar physical clocks, GentleRain has a throughput higher  than other proposals such as COPS or Orbe. However, both of its correctness and performance relay on monotonic synchronized physical clocks. Specifically, if physical clocks go backward its correctness is violated. In addition, GentleRain is sensitive on the clock synchronization, and clock skew may slow write operations in GenlteRain. In this paper, we want to solve this issue in GenlteRain by using Hybrid Logical Clock (HLC) instead of physical clocks. Using HLC, GentleRain protocl is not sensitive on the clock skew anymore. In addition, even if clocks go backward, the correctness of the system is not violated. Furthermore, by HLC, we timestamp versions with a clock very close to the physical clocks. Thus, we can take causally consistency snapshot of the system at any give physical time. 
%
We call GentleRain protocol with HLCs GentleRain+. We have implemented GentleRain+ protocol, and have evaluated it experimentally. GentleRain+ provides faster write operations compare to GentleRain that rely solely on physical clocks to achieve causal consistency. We have also shown that using HLC instead of physical clock does not have any overhead. Thus, it makes GentleRain more robust on clock anomalies at no cost. 
\end{abstract}

%% file: introduction.tex
\section{Introduction}

Distributed data stores are one of the integral parts of today's Internet services. Service providers usually replicate their data on different nodes worldwide to achieve higher performance and durability. However, when we use this approach, the consistency between replicas becomes a concern. In an ideal situation, all replicas always represent exactly the same data. In other words, any update to any data item instantaneously becomes visible in all replicas. This model of consistency is called \textit{strong consistency}. Strong consistency provides a simple semantics for programmers who want to use the distributed data store for their applications. Unfortunately, it is impossible to achieve strong consistency without sacrificing the availability when we have network partitions. In particular, in case of network partitions, to maintain the consistency between different copies of data items in different replicas, we have to make updates unavailable. The CAP theorem \cite{cap} implies that we can only choose two requirements out of strong consistency, availability, and partition-tolerance.

At the other end, the weakest consistency model is called \textit{eventual consistency} \cite{ec}. In this consistency model, as the name suggests, the only guarantee is that replicas become consistent "eventually". We can implement always-available services under this consistency model. However, it may expose clients to anomalies, and application programs need to consider such anomalies when they program their applications. To understand how eventual consistency may leads to anomaly, consider the following example from \cite{cops}: Suppose in a social network, Alice changes the status of an album to friends-only. Then, she uploads a personal photo to her album. Under eventual consistency model, the photo may become visible in a remote node before album status. In that case, a user querying a remote node may see Alice's personal photo without being her friend. Despite such anomalies, some distributed data stores use eventual consistency. One example is Dynamo \cite{dynamo} which is used by Amazon to manage the state of some of its services with very high availability requirements.

\textit{Causal consistency} is an intermediate consistency model. It is weaker than strong consistency, but stronger than eventual consistency. In particular, causal consistency requires that the effect of an event can be visible only when the effect of its causal dependencies is visible. The causal dependency captures the notion of happens-before relation define in \cite{lamport}. Under this relation, any event by a client depends on all previous events by that client. Thus, in our example, the event of uploading the photo depends on the event of changing the album status to friends-only. Thus, nobody can view Alice's photo before finding her album as friends-only. Causal consistency (with some restrictions as we explain in Section  \ref{sec:background}) is achievable with availability even in presence of networks partitions. 

COPS \cite{cops} is one of the first proposals in the community for causally consistent data stores. COPS explicitly maintains the list of dependencies for each version. Then, in a remote replica, we do not make a version visible if some of its dependencies are missing. To check weather all of dependencies of an update are present in the data center, servers communicate with each other. The problem is that the overhead of this explicit dependency checking is high, and it grows as the number of dependencies or the number of serves inside the data center grows. Orbe \cite{orbe} solves the problem of maintaining an explicit list of dependencies by using dependency matrices, but it still suffers from sending dependency check messages to other partitions. 

GentleRain \cite{gentleRain} uses a different approach in a sense that it does not send any dependency check message. Instead, it guarantees causal consistent by using physical clocks. In this way, partitions communicate with each other only \textit{periodically}, and they do not need to communicate upon receiving each new update, unlike COPS or Orbe. When compared with COPS and Orbe, this reduces the message complexity of GentleRain significantly thereby providing higher throughput. 

Although GentleRain reduces the overhead of tracking depenendencies, it relies on synchronized and monotonic clocks for both correctness and performance. Specifically, it requires that clocks are strictly increasing. This may be hard to guarantee if the underlying service such as NTP causes non-monotonic updates to POSIX time \cite{NTPbad} or suffers from leap seconds \cite{leapsecond,leapsecond2}. In addition, the clock skew between physical clocks of servers may leads to cases where GentleRain deliberately delays write operations. 

This issue is made worse if we have a federated data center. In other words, all the data in the data center is not controlled by the same entity.  Such a situation can arise when subsets of data are {\em controlled} by multiple entities. An example of such federated data center occurs when you have multiple departments in a university that contribute to the data in the data center. In this case, each partition may consist of data of individual department. A client, with appropriate authorization, may be able to access data from multiple/all departments. However, since the partitions are controlled by different entities, enforcing constraints such as tight clock synchronization is difficult. We expect this to be an important issue when the data center is obtained by data from multiple sources and while each source is willing to provide access (via GET/PUT etc) they are not willing to handover the whole data to each other.

Our goal in this paper is making GentleRain more robust on clock anomalies such as clock skew or going backward. To achieve this goal, we focus on the use of hybrid logical clocks \cite{hlc} that combine the logical clocks \cite{lamport} and physical clocks. In particular, these clocks assign a timestamp $hlc.e$ to event $e$ such that if $e$ happened before $f$ (as defined in \cite{lamport}), then $hlc.e < hlc.f$ is true. Furthermore the value of $hlc$ is very close to the physical clock and is backward compatible with NTP clocks \cite{ntp}. In particular, in \cite{hlc}, it is shown that one can replace the last 16 bits of the NTP clock (that is 64 bits long) by data necessary to model HLC. Moreover, the remaining 48 bits provide precision of microseconds. Thus, an application that uses NTP can rely on HLC without violating its correctness. 

Another aspect that we have focused on in this paper is the locality of client requests. Specifically, in GentlRain, clients can only access their local data center. Other systems like COPS also has made that assumption. We provide an impossibility results that proves this assumption is indeed required in any causally consistent data store that makes local updates visible immediately. 

{\bf Contributions of the paper. } \
\begin{itemize}
\item We show that we can improve the latency of PUT operations in GentleRain with the use of HLC. 
\item We demonstrate the efficiency provided by our approach by simulating clock skews as well as by performing experiments on cloud services provided by Amazon. 
\item We demonstrate how one can remove the reliance of GentleRain on clock synchronization service.
\item We show using HCL does not have any overhead. 
\item We provide an impossibility result that shows locality of traffic is necessary for a causally consistent data stores that makes local updates visible immediately. 
\end{itemize}

{\bf Organization of the paper. } \
In Section \ref{sec:background}, we define our system architecture and the notion of causal consistency. In Section \ref{sec:motivation}, we discuss, in detail, the issues of GentleRain that we want to address. In Section \ref{sec:hlc}, we provide a brief overview of hybrid logical clocks from \cite{hlc}. Section \ref{sec:protocol} provides our GentleRain+ and Section \ref{sec:results} provides our experimental results. We provide an impossibility result in Section \ref{sec:impossibility}. In Section \ref{sec:discussion} we discuss why HLC is preferred over physical or logical clocks. 
Finally, Section \ref{sec:conclusion} conclude the papers.

%% file: background.tex
\section{Background}
\label{sec:background}

\subsection{Architecture}
\label{sec:arch}
In this section we focus on the system architecture and assumptions that are the same as those assumed in \cite{cops}, \cite{orbe}, and \cite{gentleRain}. We consider data stores that are both partitioned and replicated. Partitioning a data stores means instead of storing all data in a single machine, storing it on several machines. Using partitioning we can have a scalable data store that is not limited by the capacity of a single machine, and can store more data by simply adding more machines. Replication, on the other hand, helps us to increase both performance and durability of our data. Specifically, by replicating data in geographically different locations, we let the clients query their local data center for reading or writing data thereby reducing the response time for client operations. In addition, by replicating data, we increase the durability of our data, as if we lose the data in some replica, we have the data in other data centers. Thus, a data store consists of $N$ partitions each of which is replicated in $M$ replicas. 

Figure \ref{fig:arch} show a typical architecture of a partitioned and replicated data store. The data is replicated in different data centers in geographically different locations, and inside each data center, data is partitioned over several servers. 
Note that the partitioning of each replica into partitions is logical. We make no assumptions about physical location of each partition. For example, in an university setting, each partition may consist of data associated with one department. All these partitions form one key-value store. Thus, each partition in the key-value store could be administered independently as long as it provides the necessary functionality of GET and PUT operations.  In Figure \ref{fig:arch}, data center $C$ is a logical data center that uses servers provided by two different departments/providers.

We assume multi-version key-value stores that store several versions for each key. A key-value store has two basic operations: $PUT(k, val)$ and $GET(k)$, where  $PUT(k, val)$ writes new version with value $val$ for item with key $k$, and  $GET(k)$ reads the value of a item with key $k$.

\begin{figure}
\begin{center}
\includegraphics[scale=0.6]{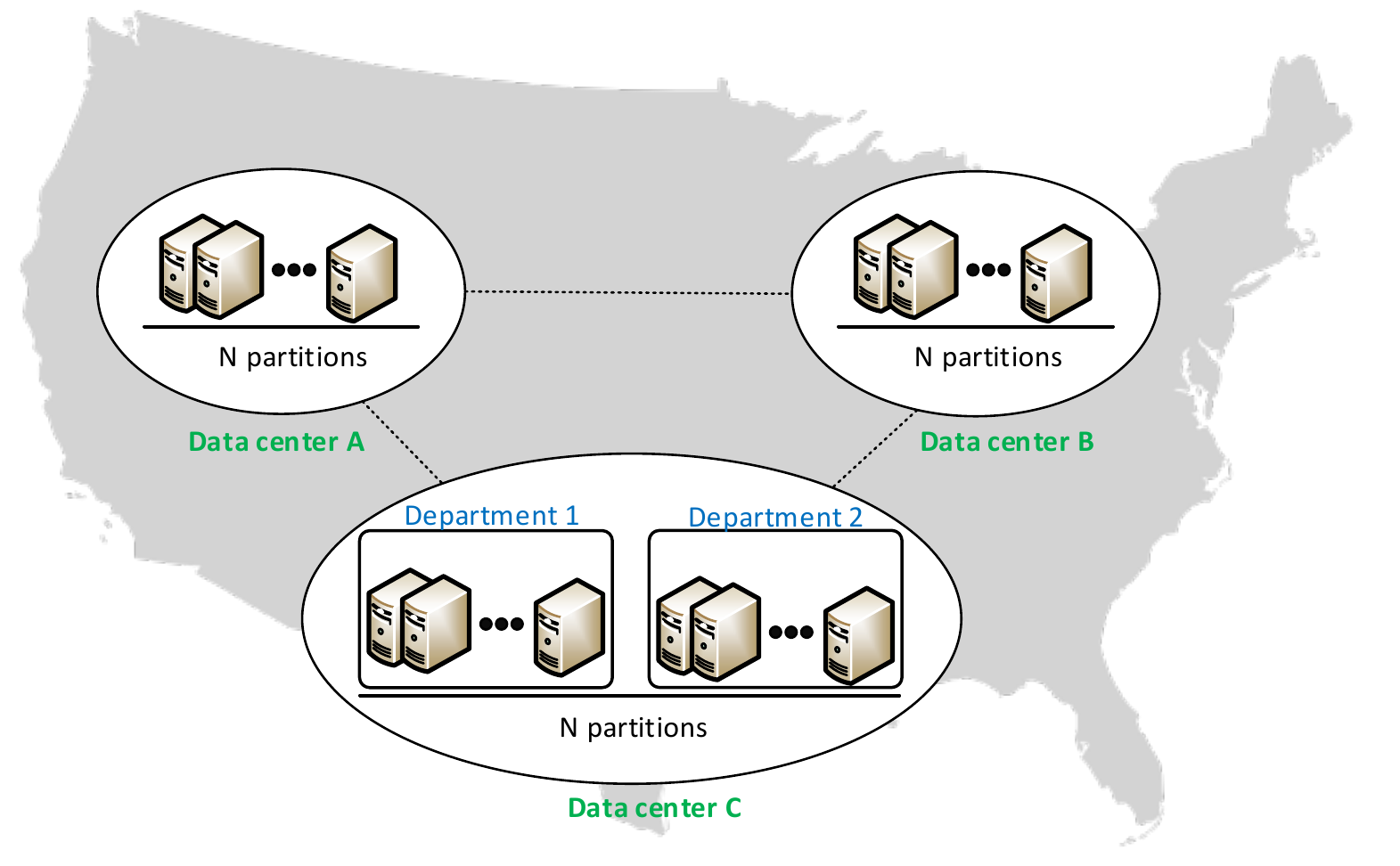}
\caption{A typical architecture of a Geo-replicated partitioned data store. Each data center is partitioned into N servers. Data center C is a logical data center consists of servers managed by different departments/providers.}
\label{fig:arch}
\end{center}
\end{figure}

\subsection{Causal Consistency}

Causal consistency is defined based on the happens-before relation between events \cite{lamport}.  In the context of key-value stores, we define happens-before relation as follows: 

\begin{definition} [Happens-before]
\label{def:happens}
Let $a$ and $b$ be two events. We say $a$ happens before $b$, and denote it as $a \rightarrow b$ iff: 
\begin{itemize}
\item $a$ and $b$ are two events by the same client (i.e., in a single thread of execution), and $a$ happens earlier than $b$, or
\item $a$ is a $PUT(k, val)$ operation and $b$ is a $GET(k)$ operation and $GET(k)$ returns the value written by $a$, or
\item there is another event $c$ such that $a \rightarrow c$ and $c \rightarrow b$. 
\end{itemize}
\end{definition}

Now, we define casual dependency as follows: 

\begin{definition} [Causal Dependency]
Let $v_1$ be a version of key $k_1$, and $v_2$ be a version of key $k_2$. We say $v_1$ causally depends on $v_2$, and denote it as $v_1 \deps v_2$ iff $PUT(k_2, v_2) \rightarrow PUT(k_1, v_1)$. 
\end{definition}



\begin{definition} [Visibility]
We say version $v$ of key $k$ is visible to a client $c$ iff $GET(k)$ performed by client $c$ returns $v'$ such that $v = v'$ or $v \rightarrow v'$.  
\end{definition}

Now, we define causal consistency as follows: 

\begin{definition} [Causal Consistency]
Let $k_1$ and $k_2$ be any two arbitrary keys in the store. Let $v_1$ be a version of key $k_1$, and $v_2$ be a version of key $k_2$ such that $v_1 \deps v_2$. The store is causally consistent if whenever $v_1$ is visible to client $c$, $v_2$ is also visible to client $c$.
\end{definition}

In the above definition we ignore the possibly of conflicts in writes. Conflicts occur when we have two writes on the same key such that there is no causal dependency relation between them. Specifically, 

\begin{definition} [Conflict]
Let $v_1$ and $v_2$ be two versions for key $k$. We call $v_1$ and $v_2$ are conflicting iff $\neg (v_1 \deps v_2)$ and $\neg (v_2 \deps v_1)$. (i.e., none of them depends on the other.)
\end{definition}

In case of conflict, we want a function that resolves the conflict. Thus, we define conflict resolution function as $f(v_1, v_2)$ that returns one of $v_1$ and $v_2$ as the winner version. Of course, if $v_1$ and $v_2$ are not conflicting, $f$ returns the latest version, i.e., if $v_1 \deps v_2$ then $f(v_1, v_2) = v_1$. For such function $f$, 
Lloyd et al. \cite{cops} have named causal consistency with conflict resolution causal+ consistency. We formalize causal+ consistency as follows: 

\begin{definition} [Causal+ Consistency]
Let $k_1$ and $k_2$ be any two arbitrary keys in the store. Let $v_1$ be a version of key $k_1$, and $v_2$ be a version of key $k_2$ such that $v_1 \deps v_2$. Let $f$ be a conflict resolution function. The store is causal+ consistent for conflict resolution function $f$ if whenever $v_1$ is visible to client $c$, then for key $k_2$, version $v_2'$ is visible to client $c$ as well where $v_2' = v_2$ or $f(v_2, v_2') = v_2'$. 
\end{definition}

A trivial solution for a causally (or causal+) consistent data store is a data store that always returns the initial value for each key. Of course, such data store is not desirable. Thus, we define causal++ consistency that requires the data store to make all local updates visible to clients immediately. Specifically, 

\begin{definition} [Causal++ Consistency]
\label{def:causal++}
A store is causal++ consistent for conflict resolution function $f$ if following conditions are hold:  

\begin{itemize}
\item it is causal+ consistent for conflict resolution function $f$, and
\item If client $c$ writes version $v_1$ for key $k$ on replica on $r$, then version $v_2$ for key $k$ is \textit{immediately} visible to any client that accesses replica $r$ where $v_2 = v_1$ or $f(v_1, v_2) = v_2$.

\end{itemize}
\end{definition}


In practice, in addition to the consistency, we want all replicas eventually converge. In other words, we want a write in a replica to reflect in other connected replicas as well. Thus, we define convergence as follows: 

\begin{definition} [Convergence] 
A store is convergent if the following condition is hold: 
 
if client $c$ writes version $v_1$ for key $k$ on replica $r$, then version $v_2$ for key $k$ is \textit{eventually} visible to any client that accesses any replica connected to $r$ where $v_2 = v_1$ or $f(v_1, v_2) = v_2$.
\end{definition}

Note that, current proposals for causally consistent data store like \cite{cops}, \cite{orbe}, and \cite{gentleRain} are infact convergent and causal++ consistent. However, they have assumed that clients do not change their replicas. In Section \ref{sec:impossibility}, we provide an impossibility result that shows achieving causal++ consistency is impossible if clients access different replicas for the same items.

In the next Section, we focus on our motivation for improving GentleRain. 

%% file: motivation.tex
\section{Motivation}
\label{sec:motivation}

Thanks to the use of clocks for dependency checking, GentleRain \cite{gentleRain} has a throughput higher than that of other causally consistent data stores proposed in the literature such as COPS \cite{cops} or Orbe \cite{orbe}. In particular, GentleRain avoids sending dependency check messages unlike other systems which leads to a light message complexity which in turn leads to higher throughput. In the following, we review the main parts of GentleRain protocol.


\textbf{GentleRain}. GentleRain assigns each version a timestamp equal to the physical clock of the partition where the write of the version occurs. We denote the timestamp assigned to version $X$ by $X.t$. GentleRain assigns timestamps such that following condition is satisfied:

\begin{itemize}
\item [$C1:$] If version $X$ of object $x$ depends on version $Y$ of object $y$, then $Y.t < X.t$.
\end{itemize} 



Also, each node in the data center periodically computes a variables called  Global Stable Time (GST) (through communication with other partitions) such that following condition is satisfied: 
\begin{itemize}
\item [$C2:$] When GST in a node has a certain value $T$, then all versions with timestamps smaller than or equal to $T$ are visible in the data center.
\end{itemize}

When a client preforms $GET(k)$, the partition storing $k$, returns the newest version $v$ of $k$ which is either created locally, or has a timestamp no greater than GST. According to conditions $C1$ and $C2$ defined above, any version which is a dependency of $v$ is visible in the local data center and causal consistency is satisfied.

\textbf{Problem: Sensitivity on physical clock and clock synchronization.} To satisfy condition $C1$, in some cases, we may have to wait before creating a new version. Specifically, in order to create a new version for a key, the client sends the timestamp $t$ of the last version that it has read/written together with the PUT operation. The partition receiving this request first waits until its physical clock is higher than $t$, and then creates the new version. The amount of this wait, as we observed in our experiments, is proportional to the clock skew between servers. In other words, as the physical clocks of servers drift from each other, the incidence and the amount of this wait period increases. This makes GentleRain sensitive to clock synchronization between servers, and a drift between clocks of the servers may have a direct negative impact on the response time for the clients operations. This problem becomes more significant when we create a federated data center by aggregating servers managed by different authorities, or even in different locations, as in this case the clock skew between servers is expected to be higher. 

In addition to the waiting issue, physical clock must be monotonic for GentleRain to be correct. In other words, if one of the physical clocks goes backward, then the correctness of GentleRain is not guaranteed, and causal consistency guarantee is violated. To understand how that can happen consider this example: consider a system consisting of two data centers $A$ and $B$. Suppose GSTs in both data centers are 6. That means, both data centers assume all versions with timestamp smaller than 6 are visible (condition $C2$). Now, suppose the physical clock of one of the servers in data center $A$ goes backward to 5. In this situation, if a client writes a new version at that server, condition $C2$ is violated, as the version with timestamp 5 has not arrived in data center $B$, but GST is 6 which is higher than 5. Since correctness of GentleRain relays on condition $C2$, by violating $C2$, correctness of GentleRain is violated. 

As we explained above, both correctness and performance of GentleRain relay on the accuracy of the physical clocks, and the clock synchronization between servers. We want to eliminate this reliance on physical clocks, and make GentleRain more robust by using HLCs. In the section \ref{sec:hlc}, we overview HLC, and then in Section \ref{sec:protocol} we focus on using HLC in GentleRain. 

%% file: hlc.tex
\section{Hybrid Logical Clocks}
\label{sec:hlc}

In this section, we recall HLCs from \cite{hlc}. HLC combines logical and physical clocks to leverage key benefits of both. We focus on the benefits of using HLCs in Section \ref{sec:discussion}.  A HLC timestamp of event $e$ has two elements: $l.e$ and $c.e$. $l.e$ is the value of the physical clock, and represents our best approximation of the global time when $e$ occurs. And, $c.e$ is a \textit{bounded} counter that is used to capture causality whenever $l.e$ is not enough to capture causality. Specifically, if we have two events $e$ and $f$ such that $e$ happens-before $f$ (see Definition \ref{def:happens}), and $l.e = l.f$, to capture causality between $e$ and $f$, we set $c.e$ to a value higher than $c.f$. Although, we increase $c$, as it is proved in \cite{hlc}, the theoretical maximum value of $c$ is $O(n)$. In practice, this value remains very small. For completeness, we recall algorithm of HLC from \cite{hlc} below. 

\begin{algorithm} 
{
\caption{HLC algorithm from \cite{hlc} }
\label{alg:hlc}
\begin{algorithmic} [1]
\STATE \textbf{Send/Local Event} 
\STATE \hspace{3mm} $l'.a = l.a$
\STATE \hspace{3mm} $l.a = max(l'.a, pt.a)$ //tracking maximum time event, $pt.a$ is physical time at a
\STATE \hspace{3mm} \textbf{if} $(l.a = l'.a) \ c.a = c.a + 1$ //tracking causality
\STATE \hspace{6mm} \textbf{else} $c.a = 0$
\STATE \hspace{3mm} Timestamp event with $l.a,c.a$

\STATE \vspace{5mm} \textbf{Receiving message $m$} 
\STATE \hspace{3mm}  $l'.a = l.a$
\STATE \hspace{3mm}  $l.a = max(l'.a, l.m, pt.a)$ //$l.m$ is $l$ value in the timestamp of the message received
\STATE \hspace{3mm}  \textbf{if} $(l.a = l'.a = l.m)$ then $c.a = max(c.a, c.m) + 1$
\STATE \hspace{3mm}  \textbf{else if} $(l.a = l'.a)$ then $c.a = c.a + 1$
\STATE \hspace{3mm}  \textbf{else if} $(l.a = l.m)$ then $c.a := c.m + 1$
\STATE \hspace{3mm}  \textbf{else} $c.a = 0$
\STATE \hspace{3mm}  Timestamp event with $l.a,c.a$
\end{algorithmic}
}
\end{algorithm}

HLC satisfies logical clock property, that it allows us to capture (one-way) causality between two events. Specifically, if $e$ happens-before $f$, then $HLC.e < HLC.f$, where $HLC.e < HLC.e$ iff $(l.e < l.f \vee ((l.e = l.f) \wedge c.e < c.f)))$.  This implies that if $HLC.e = HLC.f$, then $e$ and $f$  are (causally) concurrent. At the same time, just like physical clock, HLC increases spontaneously, and it is close to the physical clock. Thus, it can be used to take snapshot at a given physical time.  

%% file: protocol.tex
\section{GentleRain+ Protocol}
\label{sec:protocol}
GentalRain+ is basically GentleRain protocol that uses HLC instead of physical clock. For completeness of algorithms, we provide the algorithms for the basic protocol that supports PUT and GET operation. However, we assume the reader is familiar with GentleRain protocol \cite{gentleRain}. Here, we mostly focus on the use of HLC in the protocol. The protocol for transactions (GET-SNAPSHOT and GET-ROTX) is the exactly the same as that of in GentleRain. 

\subsection{Client Side}

The client side of GentleRain+ is identical to the client side of GentleRain,  except all timestamps and clock values are HLC values. Table \ref{table:clientV} shows variables maintained by a client. A client maintains a dependency time $DT_c$ that is the value of the maximum timestamp of all items accessed by the client. A client also maintains $GST_c$ that is the the maximum GST value that the client is aware of. We recall the client side algorithm from \cite{gentleRain} in Algorithm \ref{alg:client}.

\begin{center} 
\begin{table} [h]
\caption{Variables maintained by client $c$}
\label{table:clientV}
\begin{tabular}{ |l|l|l| } 

\hline
\textbf{Variable} & \textbf{Definition} \\
\hline
$DT_c$ & dependency time at client $c$ \\ 
$GST_c$ & global stable time known by client $c$ \\ 
\hline
\end{tabular}
\end{table}
\end{center}

\begin{algorithm} 
{
\caption{Client operations at client $c$ (These functions are identical to those of \cite{gentleRain} except that GST, gst, DT have HLC values instead of physical clock values) }
\label{alg:client}
\begin{algorithmic} [1]

\STATE \textbf{GET} (key $k$) 

\STATE \hspace{3mm}  send $\langle \textsc{GetReq} \ k, GST_c\rangle$  to server
\STATE \hspace{3mm}  receive $\langle \textsc{GetReply} \ v, ut, gst \rangle$
\STATE \hspace{3mm}   $DT_c \leftarrow max(DT_c,ut)$
\STATE \hspace{3mm}  $GST_c \leftarrow max(GST_c,gst)$
\RETURN $v$

\STATE \vspace{5mm} \textbf{PUT} (key $k$, value $v$)
\STATE \hspace{3mm} send $\langle \textsc{PutReq} \ k,v,DT_c \rangle$ to server
\STATE \hspace{3mm} receive $\langle \textsc{PutReply} \ ut \rangle$ 
\STATE \hspace{3mm} $DT_c \leftarrow max(DT_c,ut)$

\end{algorithmic}
}
\end{algorithm}

\subsection{Server Side}
In this section we focus on the server side of the protocol. We have $M$ replicas (i.e., $M$ data centers) each of which with $N$ partitions. We denote the $n$th partition in $m$th replica by $p^m_n$. Table \ref{table:serverV} shows variables that define the state of partition $p^m_n$. $PC^m_n$ and $HLC^m_n$ are the physical clock and the HLC of partition $p^m_n$, respectively. $VV^m_n$ is vector with size $N$ that maintains the latest (HLC) timestamps received from other replicas by partition $p^m_n$. Specifically, $VV^m_n[k]$ is the latest timestamp received from server $p^k_n$. $LST^m_n$ is minimum timestamp that $p^m_n$ has received from its peers in other data centers. In other words, $LST^m_n$ maintains the minimum value in $VV^m_n$. Servers share their $LST^m_n$ with each other, and compute GST as the minimum of LSTs. $GST^m_n$ is the GST computed in server $p^m_n$.

\begin{center}
\begin{table} [h]

\caption{Variables maintained by server $p^m_n$}
\label{table:serverV}
\begin{tabular}{ |l|l|l| } 
\hline
\textbf{Symbol} & \textbf{Definition} \\
\hline
$PC^m_n$ &  current physical clock time at $p^m_n$\\ 
$HLC^m_n$ &  current hybrid logical clock time at $p^m_n$\\ 
$VV^m_n$ & version vector of $p^m_n$ , with $M$ elements \\ 
$LST^m_n$ & local stable time at $p^m_n$ \\ 
$GST^m_n$ & global stable time at $p^m_n$ \\ 
\hline
\end{tabular}
\end{table}
\end{center}

Table \ref{table:data} shows the information that is stored for each version. For each version, in addition to the key and value, we store some additional metadata including the (HLC) time of creation of the version $ut$, and the source replica $sr$ where the version has been written.

\begin{center}
\begin{table} 
\caption{Data item}
\label{table:data}
\begin{tabular}{ |l|l|l| } 
\hline
\textbf{Symbol} & \textbf{Definition} \\
\hline
$d$ & item tuple $\langle k, v, ut, , sr \rangle$ \\
$k$ & key \\ 
$v$ & value \\
$ut$ & update time \\
$sr$ & source replica $id$ \\ 
\hline
\end{tabular}
\end{table}
\end{center}

Algorithm \ref{alg:server1} shows the first part of the server side operations of GentleRain+. Upon receiving a GET request ($\textsc{GetReq}$), the server finds the latest version of the requested key that is either written in the local data center, or its update time is less than the current GST, and then returns back this value together with its update time and GST to the client.

The main difference between GentleRain and GentleRain+ is how we process PUT requests. Upon receiving a PUT request by server $p^m_n$, the server first updates its HLC by calling function $updateHLCforPut$. This function updates the $HLC^m_n$ by the algorithm of the HLC. We pass $dt$ to $updateHLCforPut$, and $updateHLCforPut$ updates $HCL^m_n$ such that the new $HLC^m_n$ is higher than $dt$. After updating the $HLC^m_n$, the server updates its own entry in $VV^m_n$ with the $HLC^m_n$. Next, the server creates a new version for the item specified by the client and uses the current $HLC^m_n$ value for its timestamp. The server, next, sends back the assigned timestamp $d.ut$ to the client.  Note that in GentleRain, processing PUT operation is delayed deliberately if $dt$ sent by the client is higher than the current physical clock. Such delay is not needed in GentleRain+. 

Upon creating new version for an item in one data center, we should replicate the new version to other replicas in other data centers. Thus, we call $sendReplicate$ function that sends replicate messages to other data centers. Upon receiving a replicate message from server $p^k_n$, the receiving server $p^m_n$ adds the new item to the version chain of key $d.k$. The server also updates the entry for server $p^k_n$ in its version vector. Thus, it sets $VV^m_n[k]$ to $d.ut$. 

Algorithm \ref{alg:server2} shows the second part of the GentleRain+ protocol. The functions in Algorithm \ref{alg:server2} (except $updateHLC$) are the same as those of GentleRain except the use of HLC instead of physical clock. Heartbeat messages are sent by a server, if the server has not sent any replicate message for a certain time $\Delta$. The goal of heartbeat messages is updating the knowledge of the peers of an partition in other replicas (i.e., updating $VV$s). In addition, every $\theta$ time, we compute LST and GST. To compute GST, partitions needs to communicate their LST with each other. Broadcasting LSTs has a high overhead. Instead, we efficiently compute GST over a tree as describe in \cite{gentleRain}. 
\begin{algorithm} 
{
\caption{Server operations at server $p^m_n$ (part 1)}
\label{alg:server1}
\begin{algorithmic} [h]
\STATE \textbf{Upon} receive $\langle \textsc{GetReq} \ k\rangle$
\STATE \hspace{3mm} $GST^m_n \leftarrow max(GST^m_n, gst)$
\STATE \hspace{3mm}  obtain latest version $d$ from version chain of key $k$ s.t. 
\begin{itemize}
\item $d.sr = m$, or
\item $d.ut \leq GST^m_n$
\end{itemize}
\STATE \hspace{3mm} send $\langle\textsc{GetReply} \  d.v, d.ut, GST^m_n\rangle$ to client

\STATE \vspace{5mm} \textbf{Upon}  receive $\langle \textsc{PutReq} \ k, v, dt \rangle$
\STATE \hspace{3mm}  $updateHCLforPut(dt)$
\STATE \hspace{3mm}  $VV^m_n  [m] \leftarrow HLC^m_n$
\STATE \hspace{3mm}  Create new item $d$
\STATE \hspace{3mm}  $d.k \leftarrow k$ 
\STATE \hspace{3mm}  $d.v \leftarrow v$
\STATE \hspace{3mm}  $d.ut \leftarrow HLC^m_n$
\STATE \hspace{3mm}  $d.sr \leftarrow m$ 
\STATE \hspace{3mm}  insert $d$ to version chain of $k$
\STATE \hspace{3mm}  send $\langle PutReply \ d.ut \rangle$ to client
\STATE \hspace{3mm}  sendReplicates($d$)

\STATE \vspace{5mm} \textbf{Upon}  receive $\langle \textsc{Replicate} \  d\rangle$ from $p^k_n$
\STATE \hspace{3mm}  insert $d$ to version chain of key $d.k$
\STATE \hspace{3mm} $VV^m_n  [k] \leftarrow d.ut$

\STATE \vspace{5mm} \textbf{sendReplicates($d$)}
\STATE \hspace{3mm}  \textbf{for} each server $p^k_n, k \in \{0 \ldots M-1\}, k \neq m$ \textbf{do}
\STATE \hspace{6mm}  send $\langle \textsc{Replicate} \ d \rangle$ to $p^k_n$

\STATE \vspace{5mm} \textbf{updateHLCforPut ($dt$)}
\STATE \hspace{3mm} $l’ = HLC^m_n.l$
\STATE \hspace{3mm} $HLC^m_n.l = max (l', PC^m_n, dt.l)$
\STATE \hspace{3mm} \textbf{if} $(HLC^m_n.l = l' = dt.l) \ \ HLC^m_n.c \leftarrow max(HLC^m_n.c ,dt.c)+1$
\STATE \hspace{3mm} \textbf{else if} $(HLC^m_n.l = l') \  \ HLC^m_n.c \leftarrow HLC^m_n.c + 1$
\STATE \hspace{3mm} \textbf{else if} $(HLC^m_n.l = l) \  \ HLC^m_n.c \leftarrow dt.c + 1$
\STATE \hspace{3mm} \textbf{else} $HLC^m_n.c = 0$

\end{algorithmic}
}
\end{algorithm}

\begin{algorithm} 
{
\caption{Server operations at server $p^m_n$ (part 2) (These functions are identical to those of \cite{gentleRain} except we use HLC instead of physical clock)}
\label{alg:server2}
\begin{algorithmic} [h]
\STATE \textbf{Upon} every $\Delta$ time
\STATE \hspace{3mm} \textbf{if} more than $\Delta$ time unite has passed since the last heartbeat or replicate message 
\STATE \hspace{6mm}  $updateHCL()$
\STATE \hspace{6mm}  $VV^m_n  [m] \leftarrow HLC^m_n$
 \STATE \hspace{6mm}  \textbf{for} each server $p^k_n, k \in \{0 \ldots M-1\}, k \neq m$ \textbf{do}
\STATE \hspace{9mm}  send $\langle \textsc{Heartbeat} \ HLC^m_n \rangle$ to $p^k_n$

\STATE \vspace{5mm} \textbf{Upon}  receive $\langle \textsc{Heartbeat} \ hlc \rangle$ from $p^k_n$
\STATE \hspace{3mm}  $VV^m_n  [k] \leftarrow hlc$

\STATE \vspace{5mm} \textbf{Upon}  every $\theta$ time
\STATE \hspace{3mm} $LST^m_n \leftarrow min_{i=1}^{M} (VV^m_n[i])$
\STATE \hspace{3mm} $GST^m_n \leftarrow min_{j=1}^{N} (LST^m_j)$

\STATE \vspace{5mm} \textbf{updateHLC ()}
\STATE \hspace{3mm} $l’ = HLC^m_n.l$
\STATE \hspace{3mm} $HLC^m_n.l = max(HLC^m_n.l, PC^m_n)$
\STATE \hspace{3mm} \textbf{if} $(HLC^m_n.l = l’)\ $ $HLC^m_n.c \leftarrow HLC^m_n.c + 1$
\STATE \hspace{3mm} \textbf{else} $HLC^m_n.c \leftarrow 0$

\end{algorithmic}
}
\end{algorithm}

%% file: results.tex
\section{Experimental Results}
\label{sec:results}
We have implemented GenlteRain+ protocol in a distributed kay-value store called MSU-DB. MSU-DB is written in Java, and it can be downloaded from \cite{msudb}. MSU-DB uses Berkeley DB \cite{berkeleyDb} in each server for data storage and retrieval. For the server-to-server communication it relays on Netty \cite{netty}. For comparison purposes, we have implemented GentleRian in the same code base.


First, we investigate the overhead of using HLC. Next, we focus on the response time improvement for PUT operations resulted by using HLC instead of physical clocks. 
We run our experiments on Amazon AWS \cite{aws} on \texttt{c3.large} instances running Ubuntu 14.04. The specification of servers is as follows: 7 ECUs, 2 vCPUs, 2.8 GHz, Intel Xeon E5-2680v2, 3.75 GiB memory, 2 x 16 GiB Storage Capacity.

\subsection{Overhead of HLC}
In this section, we investigate the overhead of using HLC instead of physical clocks. We use the compact representation of HLC suggested in \cite{hlc}. Thus, we do not need to use two different variables for two parts of a HLC clock. Specifically, we replace the last 16 bits of the NTP clock (that is 64 bits long) by data necessary to model HLC. The remaining 48 bits provide precision of microseconds.
 
Using compact representation, HLC timestamps are that same as physical clock timestamps regarding storage and memory requirements. The only concern is the processing overhead of updating HLC timestamps for PUT operations. Although HLC update is a very simple operation, to make sure that it does not affect the throughput of the system we perform the following experiment: we run MSU-DB on a single server, and run enough clients to saturate the server with PUT operations. For each of these PUT operations, the server needs to update the HLC. We run the experiment for PUT operations with different value sizes and compute the throughput as the number of operation done in one second. 

We do the same experiment on GentleRain. Note that, running the experiment on a single server is in favor of GentleRain, as the wait period before PUT operation does not occur at all, and GentleRain can respond PUT operations immediately. We also implement the echo protocol in the same code base that simply echos the client messages. We compute the throughput of the echo server to investigate the processing power of our hardware and software platform. Table \ref{table:throughput} shows the results of our experiment. We observed that using HLC instead of physical clocks does not have any overhead on the throughput. 

\begin{center}
\begin{table} [h]
\caption{Throughput in Kop/sec for PUT operations with different value sizes.}
\label{table:throughput}
\begin{tabular}{ |l|l|l|l|l| } 
\hline
 & \textbf{Echo (16)} & \textbf{PUT (16)} & \textbf{PUT (128)} & \textbf{PUT (1K)} \\
\hline
GentleRain  & 51.612066	& 10.053134	& 7.12	& 1.5823334 \\
GentleRain+ & 51.612066 & 10.4336 & 7.0791333 & 1.5866 \\
\hline
\end{tabular}
\end{table}
\end{center}

\subsection{Response Time of PUT Operations}

In this section we study the effect of clock skew on the response time for PUT operations in GentleRain. We also investigate how HLC eliminates the negative effect of clock skew on the response time in GentleRain+.

To study the effect of clock skew on the response time accurately, we need to have a precise clock skew between servers. However, the clock skew between two different machines depends on many factors out of our control. To have a more accurate experiment, we consolidate two virtual servers on a single machine, and impose an artificial clock skew between them.  Then, we change the value of the clock skew and observe its effect of the response time for PUT operations. A client sends PUT requests to the servers in a round robin fashion. Since the physical clock of one of the servers is behind the physical clock of the other server, half of the PUT operations will be delayed by the GentleRain. On the other hand, GentleRain+ does not delay any PUT operation, and processes them immediately.  We compute the average response time for 1K PUT operations. Figure \ref{fig:cs_art} shows average response time for PUT operation in GentleRain grows linearly as the clock skew grows, while the average response time in GentleRain+ is independent of clock skew.

\begin{figure}
\begin{center}
\includegraphics[scale=0.6]{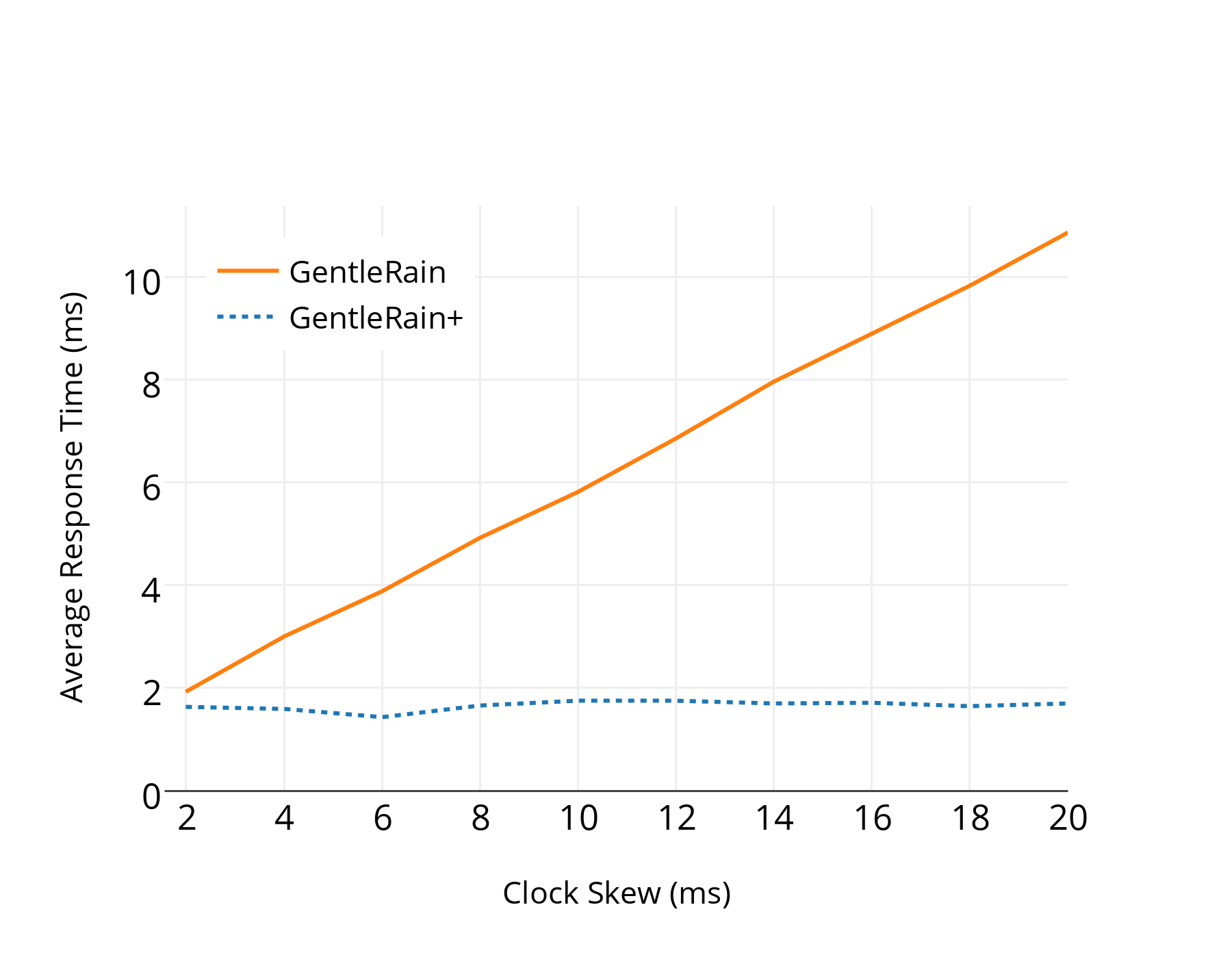}
\caption{The effect of clock skew on the average response time for PUT operations.}
\label{fig:cs_art}
\label{fig:CS_ART}
\end{center}
\end{figure}

Next, we do the same the experiment when the servers are running on two different machines. This time, we do not impose any artificial clock skew. We run NTP \cite{ntp} on servers to synchronized physical clocks. Now, the client sends PUT requests to these servers. We repeat this experiment when the serves are running in different availability zones in the Virginia region. We always run servers in the same availability zone, and the client in a different availability zone in the same region. Table \ref{table:delay} shows average delay before PUT operations in GentleRain. We do not have such delay in GentleRain+. Figure \ref{fig:aws_art} shows how this delays negatively affect the average response time of GentleRain. The increase of average response time directly affect the clients. 

\begin{center}
\begin{table}
\caption{Average delay before PUT operations in GentleRain.}
\label{table:delay}
\begin{tabular}{ |l|l|l| } 
\hline
\textbf{Availability Zone} & \textbf{Average Delay (ms)} & \textbf{Ping from client to servers (ms)} \\
\hline
\texttt{us-east-1b}  & 7.75	& 1.261 \\
\texttt{us-east-1c}  & 	12.25 & 1.973 \\
\texttt{us-east-1e}  & 	12.5 & 0.794 \\
\hline
\end{tabular}
\end{table}
\end{center}
 
 \begin{figure}
\begin{center}
\includegraphics[scale=0.6]{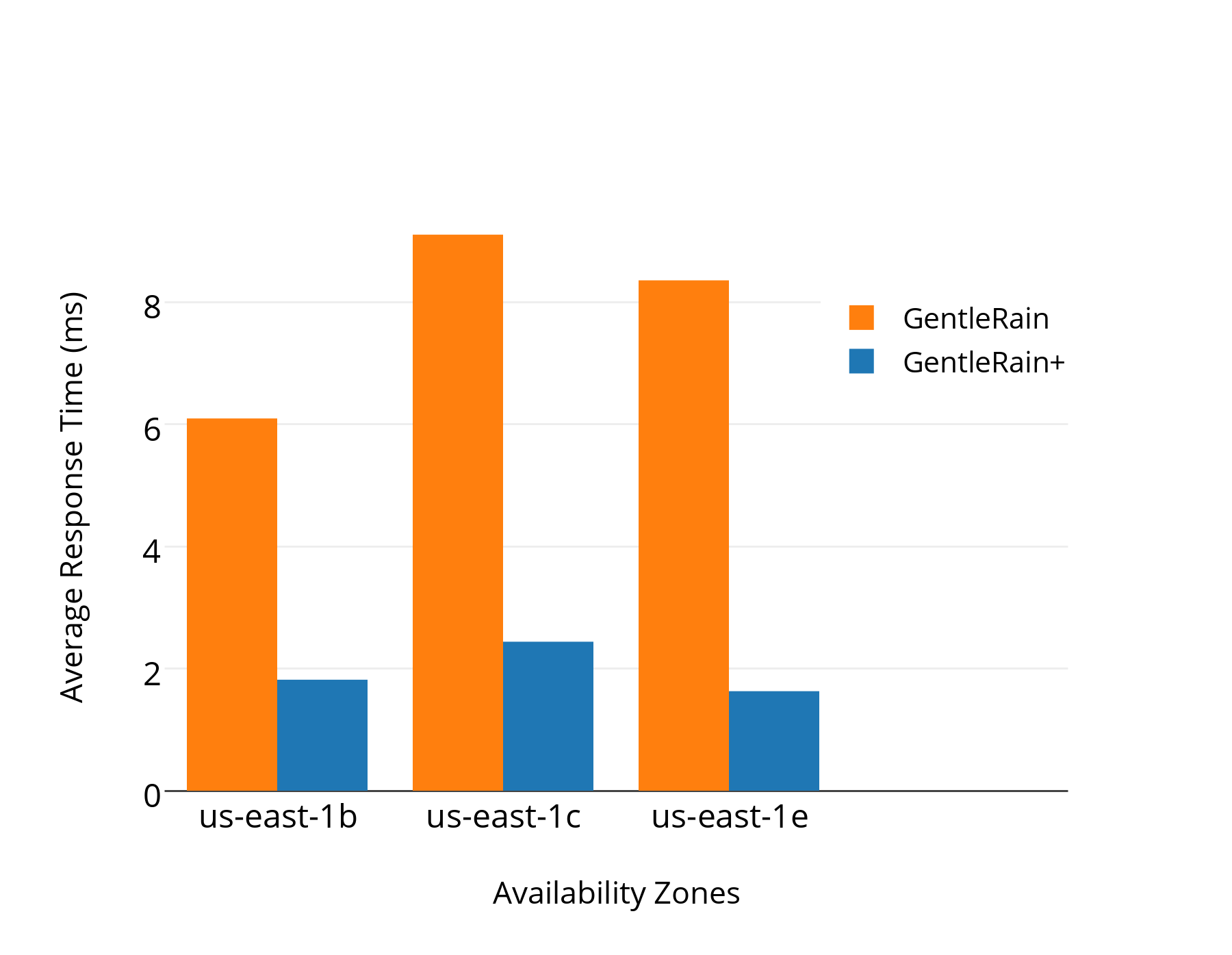}
\caption{The comparison of average response time between GentleRain and GentleRain+ in different availability zones.}
\label{fig:aws_art}
\end{center}
\end{figure}

%% file: impossibility.tex
\section{Impossibility Results}
\label{sec:impossibility}

In this section, we want to focus on the locality of traffic assumption made by GentleRain. The CAP theorem \cite{cap} proves that achieving strong consistency, together with availability is impossible under the asynchronous network model. Availability means any client operation (PUT and GET) should be answered in a finite time. Asynchronous network model means that the communication among servers can be arbitrary delayed and the network can be partitioned for an unbounded duration. 


Although achieving strong consistency together with availability is impossible in asynchronous network model, there are several proposals in the literature to guarantee causal++ consistency together with availability under asynchronous network model. However, they assume the locality of traffic. In other words, they assume that a client accesses the same replica forever. In this section, we consider the \textit{moving client model} where client can access different keys from different replicas.

 
Next, we claim that the locality of traffic assumption is essential for causal++ consistency, as we have the following impossibility result: 

\begin{theorem}
In an asynchronous network with moving client model, it is impossible to 
implement a replicated data store that guarantees following properties:

\begin{itemize}
\item Availability
\item Causal++ consistency
\end{itemize}
\end{theorem}

\begin{proof}
We prove this by contradiction. Assume an algorithm $A$ exists that guarantees availability and causal++ consistency for conflict resolution function $f$ in asynchronous network with moving client model. We create an execution of $A$ that contradicts causal++ consistency. Assume a data store where each key is stored in at least two replicas $r$ and $r'$. Initially, the data center contains two keys $k_1$ and $k_2$ with versions $v_1^0$ and $v_2^0$. These versions are replicated on $r$ and $r'$. Next the system executes in the following fashion. 


\begin{itemize}
\item There is a partition between $r$ and $r'$, i.e., all future messages between them will be lost. 
\item $c$ performs $PUT(k_1, v_1^1)$ on replica $r$.
\item $c$ performs $PUT(k_2, v_2^1)$ on replica $r$. 
\item $c'$ performs $GET(k_2)$ on replica $r$. Value returned is $v_2^1$
\item $c'$ moves to replica $r'$.
\item $c'$ performed $GET(k_1)$ on replica $r'$. Value returned is $v_1^0$.

\end{itemize}




Thus, $v_2^1$ is visible to $c'$, but its causally dependent version $v_1^1$ is not visible. One the other hand, there is no $v_3$ on replica $r'$ such that $f(v_1^1, v_3) = v_3$. Thus, causal++ consistency is violated (contradiction).
\end{proof}

%% file: discussion.tex
\section{Discussion: Why Hybrid Logical Clocks}
\label{sec:discussion}
The main property of logical clock proposed by Lamprot \cite{lamport} is that we can use logical clock to capture causality between events. Specifically, if event $e$ happens-before $f$ (see Definition \ref{def:happens}), then logical clock timestamp assigned to $e$ is smaller than logical clock timestamp assigned to $f$. Although logical clocks help us to capture causality between events, using them in the distributed data store system has some complications. First, as mentioned in \cite{gentleRain}, using logical clocks can cause clocks at different servers move at very different rates. Specifically, the logical clock of a server with higher event rate moves faster than the logical clock of a server with smaller event rate. When clocks of different servers moves at different speed, the visibility of update can be delayed arbitrary \cite{gentleRain}. Another problem of logical clocks is the timestamp assigned to version has not meaningful relation with a \textit{real clock}. Thus, if we use logical clocks for timestamping versions, we have no way to access a certain version of a key or snapshot of the system at the certain physical time. 

Physical clocks, on the other hand, do not have the above mentioned problem of logical locks. Specifically, they advance spontaneously. Thus, unlike logical clocks, even if no event occurs at a server, its clock advances by its own. Another advantage of physical clock is that if use them to timestamp versions, we can access the data of our system at a desired physical time. Although physical clocks have these benefits, they are not as efficient as logical clocks for capturing causal relation between events. In particular, if we timestamp events with physical clocks as soon as they occur, then it is not guaranteed that if event $e$ happens-before $f$ timestamp of $e$ is smaller than timestamp of $f$. To solve this problem, the approach of GentelRain is delaying write operations, but it can leads to increasing response time when clock synchronization is not good enough. Thus, it makes it sensitive on clock synchronization. 

Hybrid logical clocks combine the benefit of both logical clocks and physical clocks. Specifically, at one hand, just like physical clocks, it advances spontaneously, and it also is close to the real clock. On the other hand, without delaying event, it guarantees logical clock property that is if event $e$ happens-before $f$ timestamp of $e$ is smaller than timestamp of $f$. 

%% file: conclusion.tex
\section{Conclusion}
\label{sec:conclusion}

In the paper we provided an improvement over GentleRain protocol to make it more robust on clock anomalies. The correctness of GentleRain relays on monotonic clocks. Thus, if =clocks go backward, the correctness of GentleRain is violated. In addition, GentleRain may delay write operation, because of clock skew between servers. When clocks are synchronized with each other this delay does not occur. However, if for any reason clocks become unsynchronized, the forced delay adversely affect the response time of write operations. Thus, both correctness and performance of GentleRain relay on the synchronized monotonic clocks. Our improvement over GentleRain, name GentelRain+, solves both problems by using HLC instead of physical clock. Thus, if clocks go backward the correctness is still satisfied. In addition, clock drift does not affect the performance of the system at all, as all operations are processed immediately unlike GentleRain. 

We implemented GentelRain+ protocol in a distributed key-value store called MSU-DB, and compared it with GentleRain. As expected, GentleRain+ provides shorter response time for write operations comparing to GentleRain. We also did experiments to investigate the overhead of using HLC instead of physical clock. We observed using HLC does not have any processing overhead comparing to physical clocks. In addition, as we used compact representation of HLC, the memory and storage overhead of HLC is exactly the same as physical clock. Thus, GentleRain+ make GentleRain more robust at no cost. 

Finally, we provided an impossibility result which states in presence of network partitions the locality of traffic is necessary, to have an always available causally consistent data store that immediately makes local updates visible. 